\documentclass[journal]{IEEEtran}

\usepackage{amsthm}
\usepackage{subfig}
\usepackage[margin=0.5cm]{caption}
\usepackage{lettrine}
\usepackage{textcomp}

\usepackage{amsmath,amstext,amssymb, txfonts}
\usepackage{mathrsfs}
\usepackage{graphicx,epstopdf}
\usepackage{verbatim}
\usepackage{ifthen}
\usepackage{multirow}
\usepackage{cite}
\usepackage[usenames,dvipsnames]{color}
\usepackage{xcolor}
\usepackage{lscape}
\usepackage{mdwlist}
\usepackage{url}
\usepackage{dcolumn}
\usepackage{amsfonts}
\usepackage{latexsym}
\usepackage{url}

\usepackage{algorithm}
\usepackage{algorithmicx}
\usepackage{algpseudocode}

\newboolean{showcomments}
\setboolean{showcomments}{true}
\newcommand{\masoud}[1]{  \ifthenelse{\boolean{showcomments}}
{ \textcolor{red}{(Masoud says:  #1)}} {}  }
\newcommand{\slow}[1]{\ifthenelse{\boolean{showcomments}}
{ \textcolor{red}{(Steven says:  #1)}}{}}
\newcommand{\lijun}[1]{\ifthenelse{\boolean{showcomments}}
{ \textcolor{blue}{(Lijun says:  #1)}}{}}

\def\ba{\begin{array}}
\def\ea{\end{array}}
\newcommand{\beq}{\begin{align}}
\newcommand{\eeq}{\end{align}}
\newcommand{\bq}{\begin{eqnarray}}
\newcommand{\eq}{\end{eqnarray}}
\newcommand{\bqn}{\begin{eqnarray*}}
\newcommand{\eqn}{\end{eqnarray*}}
\newcommand{\bee}{\begin{enumerate}}
\newcommand{\eee}{\end{enumerate}}
\newcommand{\bi}{\begin{itemize}}
\newcommand{\ei}{\end{itemize}}
\newcommand{\btab}{\begin{tabular}}
\newcommand{\etab}{\end{tabular}}

\newtheorem{theorem}{Theorem}

\newtheorem{lemma}{Lemma}

\newtheorem{definition}{Definition}
\newtheorem{proposition}{Proposition}
\newcommand{\hL}{\mathcal{L}}
\newcommand{\hN}{\mathcal{N}}

\begin{document}

\title{Pseudo-gradient Based Local Voltage Control in Distribution Networks}

\author{\IEEEauthorblockN{Xinyang Zhou\IEEEauthorrefmark{2}, Masoud Farivar\IEEEauthorrefmark{1} and 
Lijun Chen\IEEEauthorrefmark{2}
\\
\IEEEauthorblockA{
\IEEEauthorblockA{\IEEEauthorrefmark{2}College of Engineering and Applied Sciences, University
of Colorado, Boulder} \\ \IEEEauthorrefmark{1}Department of Electrical Engineering,
Caltech, Pasadena}}
}

\maketitle
\thispagestyle{empty}
\pagestyle{empty}

\begin{abstract}
Voltage regulation is critical for power grids. However, it has become a much more challenging problem as distributed energy resources (DERs) such as photovoltaic and wind generators are increasingly deployed, causing rapid voltage fluctuations beyond what can be handled by the traditional voltage regulation methods. In this paper, motivated by two previously proposed inverter-based local volt/var control algorithms, we propose a pseudo-gradient based voltage control algorithm for the distribution network that does not constrain the allowable control functions and has low implementation complexity. We characterize the convergence of the proposed voltage control scheme, and compare it against the two previous algorithms in terms of the convergence condition as well as the convergence rate. \end{abstract}

\section{notation}

        \begin{tabular}{ll}   
        $t$ & time index, $t \in \mathcal{T} := \{1, 2, \ldots,\infty\}$\\
        $\hN$ & set of buses excluding bus $0$, $\hN:=\{1, ..., n\}$\\
        $\hL$ & set of power lines\\
        $\hL_i$ & set of the lines form bus 0 to bus i \\ 
        $p_i^c, q_i^c$ & real, reactive power consumption at bus $i$\\
        $p_i^g, q_i^g$& real, reactive power generation at bus $i$\\
        $P_{ij}, Q_{ij}$ & real and reactive power flow from $i$ to $j$\\
        $r_{ij}, x_{ij}$ & resistance and reactance of line $(i,j)$\\
        $V_i$ & complex voltage at bus $i$\\
        $v_i$  &    $v_i := |V_i|, ~~  i \in \hN$\\
        $I_{ij}$ & complex current from $i$ to $j$\\
        $\ell_{ij}$ & $\ell_i := |I_{ij}|^2, ~~  (i,j) \in \hL$\\
        $x^+$ &  positive part,  $ x^+ = max\left\{0,x\right\}$\\
        $[x]_a^b$ & $[x]_a^b = x + (a-x)^+ -(x-b)^+ $\\      
        $\lambda_{max}$ & the maximum eigenvalue\\
        \end{tabular}
        
 \vspace{3mm}
 A quantity without subscript is usually a vector with appropriate components defined earlier, 
e.g., $v := (v_i, i\in\hN), q^g := (q_i^g, i\in \hN)$.

\section{introduction}
Both developed and developing countries \cite{Sawin14} have been deploying large amount of renewable generations like photovoltaic (PV) and wind generators, to keep up with their ever-growing power demand, and to ease environmental problems as well. Mitigating the pressure on global environment by turning to clean energy as those distributed energy resources (DERs) are, they are doing the contrary to existing power distribution networks. Because of their innate properties of instability, renewable generations bring about rapid changes (in seconds) in voltage to the networks, considerably beyond the reach of traditional voltage regulation with capacity buses, operating at a frequency of hours.

The new IEEE Standard 1547 suggests inverter-based volt/var control in distribution systems \cite{new1547a, new1547b}. Smart inverters can be built in various sources. Take solar inverters for example: for most time, solar generations can't reach their maximum output power, and the rest available capacity can be used to absorb or generate reactive power, and inject it to the networks to regulate the voltage levels. This can be operated fast enough to keep up with the rapid voltage fluctuation and compensate it. Extensive study has been done to justify this inverter-based volt/var control\cite{chertkov11}-\cite{Sandia13}.


As a closed loop control, inverter-based volt/var control can drive voltages to desired values, by mapping current state $(v(t), q(t))$ to new reactive power injections. Former works mainly suggest two algorithms for voltage regulation, known as non-incremental algorithm and gradient algorithm. Non-incremental algorithm requires only the local voltage value to decide its reactive power injections without direct knowledge of its previous decisions \cite{Masoud-CDC13}, but oscillation problem is proved and observed, stemming from its restricted convergence condition\cite{Jahangiri13,FZC2015}. 
Gradient algorithm, carried out as an incremental algorithm, demands information of current local voltage value, as well as previous decision on reactive power injections. It exhibits better convergence properties, less restricted by parameters of networks and control functions\cite{FZC2015}.
However, implementation of gradient algorithm is difficult because it involves arduous calculation of subgradient values and inverse of control functions. Motivated by this inconvenience of gradient algorithm, we propose a pseudo-gradient algorithm, also an incremental algorithm based on both current $v(t)$ and its previous decision $q(t)$. We will characterize its convergence condition and make comparison among three different algorithms in terms of their convergence conditions and convergence rates. The comparison is conducted both analytically and numerically. The result shows that, compared with gradient algorithm, this pseudo-gradient algorithm has similar loose convergence condition, achieves a very close convergence rate, while being much easier to implement.

All the analytical characterization will be conducted based on an arbitrary radial feeder network and general control functions, and the simulations are on a distribution feeder of South California Edison, with piecewise linear control functions.

The rest of this paper is organized as follows. Section \ref{system-model} presents the system model, and briefly summarizes the major results of the non-incremental voltage control algorithm and the gradient based voltage control algorithm in \cite{Masoud-CDC13} and \cite{FZC2015}.  Section \ref{sec:pseu} presents the pseudo-gradient based voltage control algorithm and its convergence. The comparison among the three algorithms is presented in Section \ref{sec:comp}, and Section \ref{conclusion} concludes the paper.

\section{System Model}\label{system-model}
\subsection{Power flow model}
We adopt the following branch flow model \cite{Baran89, Masoud-TPS13} for a {radial} distribution system:
\begin{subequations}\label{eq:bfm}
	\begin{eqnarray}
	P_{ij} &=& p_j^c - p_j^g + \sum_{k: (j,k)\in \hL} P_{jk}+  r_{ij}  \ell_{ij}  \label{p_balance}, \\
	Q_{ij} &=&  q_j^c-q_j^g + \sum_{k: (j,k)\in \hL} Q_{jk} + x_{ij} \ell_{ij} \label{q_balance},\\
	v_j^2 &=&  v_i^2 - 2 \left(r_{ij} P_{ij} + x_{ij} Q_{ij}\right) + \left(r_{ij}^2+x_{ij}^2\right) \ell_{ij} \label{v_drop},\\
	\ell_{ij}v_i &=&   P_{ij}^2 + Q_{ij}^2  \label{currents}.
	\end{eqnarray}
\end{subequations}
Following  \cite{Baran89-2}, \cite{Masoud-CDC13}, we use a linearized version of the above model by letting  $\ell_{ij}=0$ for all $(i,j) \in \hL$ in \eqref{eq:bfm}. 
This approximation neglects the higher order real and reactive power loss terms.
Since losses are typically much smaller than power flows $P_{ij}$ and $Q_{ij}$, this only
introduces  a small relative error, typically on the order of $1\%$ \cite{Baran89}. We further assume that $v_i \approx 1$ so that we can set $v_j^2 - v_i^2 = 2 (v_j - v_i)$
in equation \eqref{v_drop}.\footnote{Note that this assumption is not essential and we can also work with $v_i^2$ instead.} This approximation introduces a small relative error of at most $0.25\%$ ($1\%$) if there is a $5\%$ ($10\%$) deviation in voltage magnitude \cite{Masoud-CDC13}.
With the above approximations, the power flow model \eqref{eq:bfm} simplifies to the following linear model:

\begin{equation*}  
v = \overline{v}_0 + R (p^g - p^c) + X(q^g - q^c),
\end{equation*}
where $\overline{v}_0 = (v_0, \dots, v_0)$ is an $n$-dimensional vector, and { resistance matrix} $R=[R_{ij}]_{n\times n}$ and { reactance matrix} $X=[X_{ij}]_{n\times n}$ 
are symmetric matrices with entries
 
\begin{eqnarray}
R_{ij}:=  \sum_{(h,k)\in \hL_i \cap \hL_j}r_{hk},
\ \ \ \ X_{ij}:=   \sum_{(h,k)\in \hL_i \cap \hL_j}x_{hk}  \label{X_def}.
\end{eqnarray}

In this paper we assume that $\overline{v}_0, p^c, p^g, q^c$ are given constants.  The only 
variables are (column) vectors $v := (v_1, \dots, v_n)$ of squared voltage magnitudes
 and $q^g := (q^g_1, \dots, q^g_n)$ of reactive powers.
 Let  $\tilde{v} = \overline{v}_0 + R (p^g - p^c) - X q^c$,  which is a constant vector.
 For notational simplicity in the rest of the paper we will ignore the superscript in $q^g$ and write $q$
 instead.  Then the linearized branch flow model reduces to the following simple form:
\begin{equation}  \label{model_2}
v = Xq + \tilde{v}.
\end{equation}
It has been shown in  \cite{Masoud-CDC13} that the matrix $X$ is positive definite.

\subsection{Local volt/var control}

The goal of volt/var control on a distribution network is to provision reactive
power injections $q := (q_1, \dots, q_n)$ in order to maintain the bus voltages 
$v := (v_1, \dots, v_n)$ within a tight range around their nominal values 
$v^{\text{nom}}_i, ~i\in\hN$.
This can be modeled by a feedback dynamical system with state $(v(t), q(t))$ at discrete
time $t$.   A general volt/var control algorithm maps the current state $(v(t), q(t))$
to a new reactive power injections $q(t+1)$. 
The new $q(t+1)$ produces a new
voltage magnitudes $v(t+1)$ according to  \eqref{model_2}. Usually $q(t+1)$ is dertermined either completely or partly according to a certain volt/var control function defined as follows:

%



\begin{definition}
A volt/var control function $f : \mathbb{R}^{n} \rightarrow  \mathbb{R}^{n}$ is a collection of 
$f_i: \mathbb{R}\rightarrow\mathbb{R}$ functions, each of which maps the current local voltage 
$v_i$ to a local control variable $o_i$ in reactive power at bus $i$:
\begin{eqnarray} 
o_i  \ = \ f_i (v_i),\quad \forall i\in \hN.  \label{eq:ec}
\end{eqnarray}
\end{definition}
The control functions $f_i$ are usually decreasing but not always strictly decreasing because of the deadband in control, as well as the bounds of feasible reactive power injections. We assume for each bus $i\in\mathcal{N}$ a symmetric deadband around the nominal voltage $(v_i^{nom}-\delta_i/2, v_i^{nom}+\delta_i/2)$ with $\delta_i\geq 0$. 
We make the following two assumptions \cite{Masoud-CDC13}:
\begin{enumerate}
	\item [A1:]\label{A1} The volt/var control functions $f_i$ are non-increasing over $\mathcal{R}$ and strictly decreasing and differentiable in $(-\infty, v_i^{nom}-\delta/2)$ and in $(v_i^{nom}+\delta_i/2, \infty)$.
	\item [A2:]\label{A2} The derivative of the control function $f_i$ is bounded, i.e., there exists a finite $\alpha_i$ such that $|f_i^{\prime}(v_i)|\leq\alpha_i$ for all $v_i$ in the appropriate domain.
\end{enumerate}

As an illustrative example, see Fig. \ref{inverse} (left) for the piecewise linear droop control function proposed in the latest draft of IEEE 1547.8 Standard\cite{new1547a}:\footnote{Here we also use $\alpha_i$ to indicate the slope of the droop control function, which does not contradict the use of $\alpha_i$ in the condition A2.}
\begin{align} \label{eq:plf}
f_i(v_i) :=   \left[-\alpha_i \left( v_i - v_i^{\text{nom}} - \frac{\delta_i}{2}\right)^{+} 
+ \alpha_i \left( - v_i + v_i^{\text{nom}}- \frac{\delta_i}{2} \right)^{+} 
\right]^{{q_i}^{\max}} _{{q_i}^{\min}},
\end{align}
where the local control variable in reactive power is constrained to within $[{q_i}^{\min}, {q_i}^{\max}]$. This particular control function will be used in the numerical examples presented in Section \ref{sec:numerical}. 
\begin{figure*}
	\centering
	\includegraphics[scale=0.4]{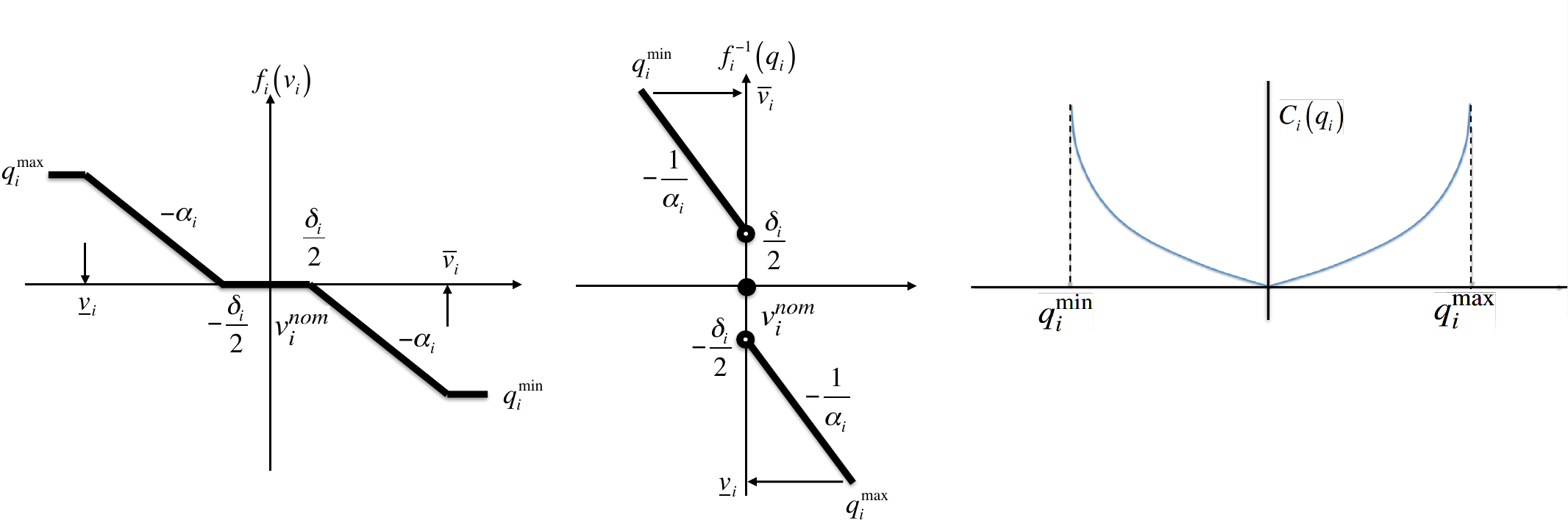}
	\caption{From left to right: piecewise linear volt/var control curve $f_i$ discussed in the draft of the IEEE 1547 standard  \cite{new1547a}, its inverse $f_i^{-1}$, and the corresponding reverse-engineered cost function $C_i$.}
	\label{inverse}
\end{figure*}

\subsubsection{Non-incremental control algorithm}
Motivated by IEEE Standard 1547 \cite{new1547a}\cite{new1547b}, we have studied in \cite{Masoud-CDC13} a local volt/var control where each
bus $i$ makes an individual decision $q_i (t + 1)$ based only on its own voltage $v_i (t)$, i.e., $q_i(t+1)=o_i(t)$, which we call {\em non-incremental control} as the current decision on reactive power injection does not depend directly on the decision at the previous time. We thus obtain the following dynamical  system that models the non-incremental local volt/var control in 
the distribution network:
\begin{eqnarray}
D1: \left |
  \begin{array}{l l l}
    v(t) &= &  Xq(t) + \tilde{v},\\
    q(t+1) &= & f(v(t)).
  \end{array} \right.
  \label{eq:dynamic1}
\end{eqnarray}

\begin{definition}
	$(v^*,q^*)$ is called an equilibrium point for $D1$, if
	\begin{eqnarray}
	v^* & = & Xq^*+\tilde{v},\nonumber\\
	q^* & = & f(v^*).\nonumber
	\end{eqnarray}
\end{definition}
By reverse engineering, we have shown in \cite{Masoud-CDC13} that the dynamical system $D1$ can be seen as a distributed optimization algorithm for solving a well-defined optimization problem
\begin{eqnarray}
&\underset{q\in\Omega}{\min}&F(q):=C(q)+\frac{1}{2}q^TXq+q^T\tilde{v}\label{eq:Fq},
\end{eqnarray}
where $C(q)=\sum_{i\in\mathcal{N}}C_i(q_i)$ with the convex cost function for each bus $i\in\mathcal{N}$ defined by $C_i(q_i):=-\int_{0}^{q_i}f_i^{-1}(q)dq$:
\begin{itemize}
\item There exists a unique equilibrium point $(v^*,q^*)$ under the condition $A1$; and $(v^*,q^*)$ is an equilibrium point if and only if $q^*$ is the unique optimal solution of (\ref{eq:Fq}). 
\item Under the conditions A1-A2, the dynamical system $D1$ converges to the unique equilibrium if the following condition $C1$ holds: 
\begin{eqnarray}
C1:~~ A^{-1} \succ X,
\end{eqnarray}
where $A^{-1}:=\mbox{diag}(\frac{1}{\alpha_i})$.
\end{itemize}

As mentioned in \cite{FZC2015}, the condition $C1$ is hard to verify in practice. First, it is a computationally demanding problem to verify a linear matrix inequality of potentially very large dimension. Second, matrix $X$ depends on the reactance of every line in the network, which is practically hard to obtain. Moreover,  $C1$ is rather restrictive in constraining allowable control functions, and the existing control schemes may not satisfy this condition. We therefore propose in \cite{FZC2015} an incremental voltage control based on the (sub)gradient algorithm for solving the optimization problem (\ref{eq:Fq}), which leads to a local var/volt control scheme that demands less restrictive condition for the control functions.

\subsubsection{Incremental control based on the gradient algorithm}


%
By applying (sub)gradient algorithm to the optimization problem (\ref{eq:Fq}), we obtain the following dynamical system with an incremental local volt/var control for 
the distribution network:

\begin{eqnarray}
D2: \left | 
  \begin{array}{l l l }
    v(t) &= &  Xq(t) + \tilde{v},\\
     q_i(t+1) & = & \left[ q_i(t) - \gamma_g \frac{\partial F(q)}{\partial q_i}\right]_{q_i^{min}}^{q_i^{max}},
  \end{array} \right. \label{eq:dynamic2}
\end{eqnarray}
where $\gamma_g>0$ is the stepsize, and the (sub)gradient is calculated by
\begin{equation} \frac{\partial F(q)}{\partial q_i} = \left\{ \begin{array}{ll}
        C_i^{'} (q_i(t)) + v_i(t) & \mbox{if ~$q_i(t) \neq 0$;}\\
        0 & \mbox{if ~$q_i(t) = 0 ~,~ -\frac{\delta}{2} \leq v_i(t) \leq \frac{\delta}{2}$;}\\
        -\frac{\delta}{2} + v_i(t) & \mbox{if ~$q_i(t) = 0 ~,~ v_i(t) > \frac{\delta}{2}$;} \\
        \frac{\delta}{2} + v_i(t) & \mbox{if ~$q_i(t) = 0 ~,~  v_i(t) <  -\frac{\delta}{2}$.} \end{array} \right. \label{eq:sg}      
\end{equation}

%
We have the following result \cite{FZC2015}: 
\begin{itemize}
\item Under the condition $A1$, the dynamical systems $D1$ and $D2$ have the same, unique equilibrium point $(v^*,q^*)$; and $(v^*,q^*)$ is an equilibrium point if and only if $q^*$ is the unique optimal solution of (\ref{eq:Fq}). 
\item Under the condition A1, the dynamical system $D2$ converges to the unique equilibrium if the following condition $C2$ on the stepsize holds: 
\begin{eqnarray}
C2:~~ \gamma_g < \frac{2}{\lambda_{max}(\nabla^2C(q)+X)},
\end{eqnarray}
where $\lambda_{max}$ denotes the maximum eigenvalue.
\end{itemize}


Compared with $C1$, $C2$ is a much less restrictive. No matter what the reactance matrix $X$ is and no matter what the control function $f_i$ is (as long as it satisfies the condition A1), we can always find an appropriate stepsize $\gamma_g$ such that $D2$ converges to its unique equilibrium, which solves the optimization problem (\ref{eq:Fq}).

\subsection{Motivation for new control algorithm}  \label{sect:mot}
Despite the condition $C2$ being less restrictive, the above incremental voltage control based on the (sub)gradient algorithm incurs lots of implementation complexity. The (sub)gradient \eqref{eq:sg} requires tracking the value of $v_i$ with respect to $\pm \delta_i/2$, and takes different forms accordingly. Furthermore, it requires the computation of the inverse of the control function $f_i$, which is computationally expensive for a general control function. This high implementation complexity of the gradient algorithm motivates us to seek an incremental voltage control algorithm with less restrictive condition on the control function as well as low implementation complexity. In the next section, we will present such a control algorithm based on the pseudo-gradient algorithm for the optimization problem (\ref{eq:Fq}) and study its equilibrium and dynamical properties. 


\section{pseudo-gradient based local voltage control}\label{sec:pseu}

Consider the following incremental local voltage control based on the pseudo-gradient algorithm for solving the optimization problem \eqref{eq:Fq}:  
\begin{eqnarray}
\nonumber q_i(t+1) &=&  [(1-\gamma_p)q_i(t)+\gamma_p f_i(v_i(t))]_{q_i^{min}}^{q_i^{max}}\nonumber\\
&=&\left[ q_i(t) - \gamma_p \Big(q_i(t)- f_i(v_i(t))\Big)\right]_{q_i^{min}}^{q_i^{max}}, \label{eq:pga}
\end{eqnarray}
where $\gamma_p >0$ is the stepsize or the weight. With the given control functions $f_i$, the implementation of the algorithm \eqref{eq:pga} is straightforward and does not have any implementation issues that the (sub)gradient based control algorithm in \eqref{eq:dynamic2} has; see the discussion in Section \ref{sect:mot}. It is also interesting to notice that, when the weight $\gamma_p=1$, we recover the non-incremental voltage control in \eqref{eq:dynamic1}.

With the control \eqref{eq:pga}, we obtain the following dynamical system:
\begin{eqnarray}
D3: \left | 
  \begin{array}{l l l }
    v(t) &= &  Xq(t) + \tilde{v},\\
     q_i(t+1) & = & \left[ q_i(t) - \gamma_p \Big(q_i(t)- f_i(v_i(t))\Big)\right]_{q_i^{min}}^{q_i^{max}}.
  \end{array} \right.
\end{eqnarray}
The dynamical system $D3$ has the same equilibrium condition as the dynamical system $D1$. 
The following result is immediate.
\begin{theorem}
\label{thm:eqnpg}
Suppose A1 holds.  There exists a unique equilibrium point for the dynamical system $D3$. 
Moreover, a point $(v^*, q^*)$ is an equilibrium if and only if $q^*$ is the unique
optimal solution of problem \eqref{eq:Fq} and $v^* = Xq^* + \tilde{v}$.
\end{theorem} 

We now analyze the convergence of the dynamical system $D3$. 
\begin{lemma}\label{lemma1}
	Suppose $A1-A2$ hold. With any $q_a,q_b\in[q_i^{min},q_i^{max}]$, we have
	\begin{eqnarray}
		\Big((-f_i^{-1}(q_a))-(-f_i^{-1}(q_b))\Big)(q_a-q_b)\geq \frac{1}{\alpha_i}(q_a-q_b)^2.\label{eq:ass12}
	\end{eqnarray}
\end{lemma}
\begin{proof}
By the condition $A2$, we have the bound on the derivative of the control function $|f_i^{\prime}(v_i)|\leq\alpha_i$, and thus the bound for its inverse $|(-f^{-1}_i(q_i))^{\prime}|=|\frac{1}{f^{\prime}(f^{-1}(q_i))}|\geq\frac{1}{\alpha_i}$.

If $q_a$ and $q_b$ are both positive (or both negative), then the corresponding $v_a=f_i^{-1}(q_a)$ and $v_b=f_i^{-1}(q_b)$ are both smaller than $v_i^{nom}-\delta_i/2$ (or larger than $v_i^{nom}+\delta_i/2$). We thus have $|(-f_i^{-1}(q_a))-(-f_i^{-1}(q_b))|\geq \frac{1}{\alpha_i}|q_a-q_b|$. Equality is achieved if the linear control function (\ref{eq:plf}) is used. 
On the other hand, if one of $q_a$ and $q_b$ is positive and the other is negative, then as long as $\delta\neq 0$ we have $|(-f_i^{-1}(q_a))-(-f_i^{-1}(q_b))|> \frac{1}{\alpha_i}|q_a-q_b|$.
Combined with the monotonicity of $f^{-1}$, the inequality (\ref{eq:ass12}) follows.
\end{proof}

\begin{theorem}\label{thm:mcon2}
	Suppose A1-A2 hold. If the stepsize $\gamma_p$ satisfies the following condition $C3$: 
	\begin{eqnarray}
	C3:~~ \gamma_p < \frac{2}{\max\{\alpha_i\}\lambda_{\mbox{max}}(\nabla^2 C(q)+X)}, \label{eq:convcon2}
	\end{eqnarray}
	then the dynamical system $D3$ converges to its unique equilibrium.
\end{theorem}

\begin{proof}
We first consider the case when $q_i(t)\neq 0,\:\forall i$, i.e.,  when objective function $F$ is differentiable. By the second-order Taylor expansion, we have
\begin{eqnarray}
&&F(q(t+1))\nonumber\\
&=& F(q(t))-\gamma_p\sum_i (-f^{-1}_i(q_i(t))+v_i(t))(q_i(t)-f_i(v_i(t)))\nonumber\\
&&+\frac{\gamma_p^2}{2}(q(t)-f(v(t)))^T(\nabla^2C(\tilde{q})+X)(q(t)-f(v(t))),\nonumber\\
\end{eqnarray}
where $f(v(t)):=\Big(f_1(v_1(T)),\ldots,f_n(v_n(t))\Big)^T$, and $\tilde{q}=\theta q(t)+(1-\theta)q(t+1)$ for some $\theta\in [0,1]$.
By Lemma \ref{lemma1}, we have  $(-f^{-1}_i(q_i(t))+v_i(t))(q_i(t)-f_i(v_i(t)))\geq \frac{1}{\alpha_i}(q_i(t)-f_i(v_i(t)))^2$. 
Thus the Taylor expansion follows as
\begin{eqnarray}
&&F(q(t+1))\nonumber\\
&\leq & F(q(t))+\frac{1}{2}(q(t)-f(v(t)))^T\nonumber\\
&&(\gamma_p^2(\nabla^2C(\tilde{q})+X)-2\gamma_p A^{-1})(q(t)-f(v(t))).\nonumber\\\label{eq:conlya}
\end{eqnarray}
When the condition $C3$ holds, $\gamma^2(\nabla^2C(\tilde{q})+X)-2\gamma A^{-1}$ is always negative definite. As a result, the second term in (\ref{eq:conlya}) is always non-positive. In fact, this part is equal to zero if and only if $q(t)=f(v(t))$, or equivalently, $q(t)=q(t+1)$. Therefore $F(q(t+1))\leq F(q(t))$, where the equality is obtained if and only if $q(t+1)=q(t)$. Besides, because of the uniqueness of the equilibrium point as shown in Theorem \ref{thm:eqnpg}, $F(q(t+1))= F(q(t))$ if and only if $q(t+1)=q(t)=q^*$. So, $F$ can be seen as a discrete-time Lyapunov function for the dynamical system $D3$, and by the Lyapunov stability theorem, the equilibrium $q^*$ is globally asymptotically stable.
 
Next we consider the case when $q_i(t)=0$ for some $i$. For bus $i$ with $q_i(t)=0$, the dynamics, irrelevant of derivative, is still well-defined, giving $q_i(t+1)=\gamma f_i(v_i(t))=0$. However, its Taylor expansion involves the derivative of $C_i(q(t))$, which doesn't exist at $q_i=0$. We thus assign subgradient value for bus $i$ as $\frac{\partial F(q)}{\partial q_i}\big|_{q_i=0} =0$, and then the proof follows similarly with this well-defined Taylor expansion, and the conclusion holds as well.
\end{proof}

Theorem \ref{thm:mcon2} shows that the pseudo-gradient based local voltage control has the same advantage as the gradient based control, as opposed to the nonincremental voltage control; and in particular, its convergence condition does not restrict the allowable control functions $f_i$. 
We will provide more detailed comparison between the three algorithms in the next section. 

{\em Remarks:}~Notice that in the pseudo-gradient algorithm it is usually assumed that  $\gamma_p\leq 1$. This gives a nice interpretation of the new decision $q_i(t+1)$ being a convex combination of the previous decision $q_i(t)$ and the local control $o_i(t)=f_i(v(t))$ in reactive power. However, here we do not require $\gamma_p\leq 1$, as long as the condition $C3$ is met.

\section{Comparative Study of Convergence Conditions and Rates}\label{sec:comp}
We have presented three different local voltage control algorithms in the previous two sections. In this section, we compare these three control schemes regarding the corresponding convergence conditions and convergence rates. As we will see, the gradient and pseudo-gradient based algorithms have very close performance in terms of convergence. 
So, as discussed in the previous sections, the advantage of the pseudo-gradient based algorithm over the gradient based algorithm is its much lower implementation complexity. However, this low implementation complexity provides strong enough motivation for adopting the pseudo-gradient based local voltage control in the distribution network. 

%
\subsection{Analytical characterization}
We start with showing the relationship between the dynamical systems $D1$ and $D3$. The following result is immediate. 
\begin{proposition}\label{prop1}
The non-incremental voltage control in the dynamical system $D1$ is a special case of the control in $D3$ with the stepsize $\gamma_p=1$. 
\end{proposition}

As a result of Proposition \ref{prop1}, when the condition $C1$ holds, the largest stepsize that $D3$ can take is no smaller than 1. On the other hand, if the condition $C3$ gives a upper bound for $\gamma_p$ that is smaller than 1, $D1$ will not converge. 

Next we investigate the relationship between the dynamical systems $D2$ and $D3$, in terms of the available ranges of the step sizes $\gamma_g$ and $\gamma_p$ for convergence and the convergence speed by looking at the largest decrease in the objective value they can make.

\begin{proposition}\label{prop2}
	The dynamical systems $D2$ and $D3$ have same (one-to-one corresponding) ranges for the step sizes $\gamma_g$ and $\gamma_p$ for convergence, i.e., for any $\gamma_p\in(0,B_p]$, there exists a corresponding $\gamma_g=\max\{a_i\}\gamma_p\in(0,B_g]$, where $B_g$ and $B_p$ are upper bounds on the stepsize for $D2$ and $D3$ to converge.
\end{proposition}
\begin{proof}
	The result follows from the proofs of the sufficiency of the conditions $C2$ and $C3$ for convergence. 
\end{proof}

%
	
Based on the proofs of the sufficiency of the conditions $C2$ and $C3$ for convergence, we also expect the similar convergence speed to the equilibrium of the dynamical systems $D2$ and $D3$. Although it is difficult to compare the exact descent rates between them, we can compare the largest decreases in the objective value that are given by the second order Taylor expansions. By the second-order Taylor expansion, the descent of the pseudo-gradient algorithm is upper bounded by 
$$\left|\Big(q(t)-f(v(t))\Big)^T\Big(\gamma_p^2(\nabla C(\tilde{q})+X)-2\gamma_pA^{-1}\Big)\Big(q(t)-f(v(t))\Big)\right|,$$ while that of the gradient algorithm is by 
$$\left|\Big(v(t)-f^{-1}(q(t))\Big)^T\Big(\gamma_g^2(\nabla C(\tilde{q})+X)-2\gamma_pI\Big)\Big(v(t)-f^{-1}(q(t))\Big)\right|.$$
Notice that from the proof of Lemma \ref{lemma1}, there exists a factor of $\max\{\alpha_i\}$ between $(q(t)-f(v(t)))$ and $(v(t)-f^{-1}(q(t)))$. On the other hand, by Proposition \ref{prop2} there is a ``compensating'' factor $1/\max\{a_i\}$ from the one-to-one correspondence between the stepsizes $\gamma_p$ and $\gamma_g$. As a result, the above two decent terms are approximately the same with the stepsizes carefully chosen.

%
%

\subsection{Numerical examples}\label{sec:numerical}

 \begin{figure*}[t]
\centering
\includegraphics[scale=0.45]{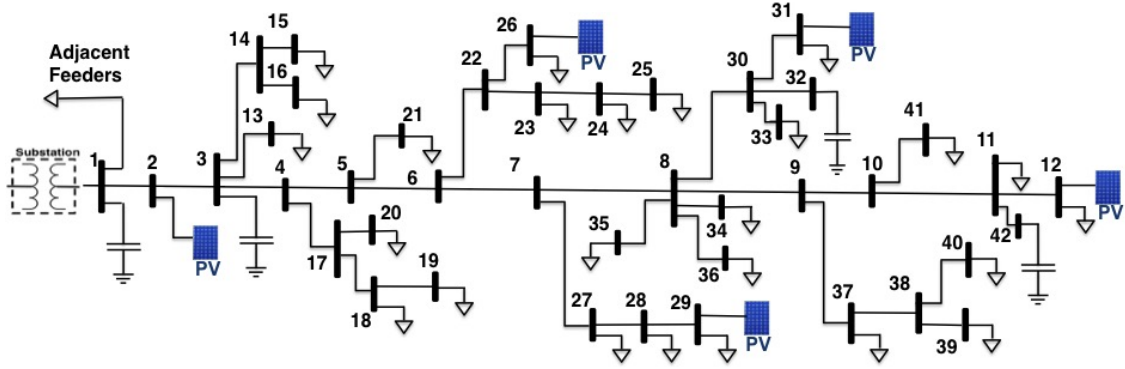}
\caption{Circuit diagram for SCE distribution system.}
 \label{42bus}
\end{figure*}

\begin{table*}
\caption{Network of Fig.~\ref{42bus}: Line impedances, peak spot load KVA, Capacitors and PV generation's nameplate ratings. }
\centering
\begin{tabular}{|c|c|c|c|c|c|c|c|c|c|c|c|c|c|c|c|c|c|}
\hline
\multicolumn{18}{|c|}{Network Data}\\
\hline
\multicolumn{4}{|c|}{Line Data}& \multicolumn{4}{|c|}{Line Data}& \multicolumn{4}{|c|}{Line Data}& \multicolumn{2}{|c|}{Load Data}& \multicolumn{2}{|c|}{Load Data}&\multicolumn{2}{|c|}{PV Generators}\\
\hline
From&To&R&X&From&To& R& X& From& To& R& X& Bus& Peak & Bus& Peak & Bus&Capacity \\
Bus.&Bus.&$(\Omega)$& $(\Omega)$ & Bus. & Bus. & $(\Omega)$ & $(\Omega)$ & Bus.& Bus.& $(\Omega)$ & $(\Omega)$ & No.&  MVA& No.& MVA& No.& MW\\
\hline
1	&	2	&	0.259	&	0.808	&	8	&	34	&	0.244	&	0.046 	&	18	&	19	&	0.198	&	 0.046	&	11	&	 0.67	&	28	&	 0.27 	&				 &		 \\

2	&	3	&	0.031	&	0.092	&	8	&	36	&	0.107	 &	0.031 	&	22	&	26	&	0.046	&	 0.015	&	12	&	0.45		 &	29	&	0.2 &			 2	&	 1\\

3	&	4	&	0.046	&	0.092	&	8	&	30	&	0.076	 &	0.015 	&	22	&	23	&	0.107	&	 0.031	&	13	&	0.89	 &	31	&	0.27	  &			26	&	 2	 \\

3	&	13	&	0.092	&	0.031	&	8	&	9	&	0.031	 &	0.031 	&	23	&	24	&	0.107	&	 0.031	&	15	&	0.07	 &	33	&	0.45	 	 &			29	&	 1.8 	 \\

3	&	14	&	0.214	&	0.046	&	9	&	10	&	0.015	 &	0.015	&	24	&	25	&	0.061	&	 0.015	&	16	&	0.67	 &	34	&	1.34  &			31	&	 2.5 	 \\

4	&	17	&	0.336	&	0.061	&	 9	&	37	&	0.153	 &	0.046 	&	27	&	28	&	0.046	&	 0.015	&	18	&	0.45	 &	35	&	0.13	  &			12	&	 3 	 \\

4	&	5	&	0.107	&	0.183	&	10	&	11	&	0.107	 &	0.076 	&	28	&	29	&	0.031	&	0		&	19	&	1.23	 &	36	&	0.67	  &			&		 \\

5	&	21	&	0.061	&	0.015	&	10	&	41	&	0.229	 &	0.122 	&	30	&	31	&	0.076	&	 0.015  &	20	&	0.45	 &	37	&	0.13	 	&			 &			 \\

5	&	6	&	0.015	&	0.031	&	11	&	42	&	0.031	 &	0.015 	&	30	&	32	&	0.076	&	 0.046	&	21	&	0.2 &	39	&	0.45 	&			  &	 \\

6	&	22	&	0.168	&	0.061	&	11	&	12	&	0.076	 &	0.046 	&	38	&	39	&	0.107	&	 0.015	&	23	&	0.13		 &	40	&	0.2 	&		 &	\\

6	&	7	&	0.031	&	0.046	&	14	&	16	&	0.046	 &	0.015 	&	38	&	40	&	0.061	&	0.015	&	24	&	0.13	 &	 41	&	0.45		&		&		 \\		\cline{15-18}

7	&	27	&	0.076	&	0.015	&	14	&	15	&	0.107	 &	0.015	&	43	&	44	&	0.061	&	0.015	&	 25	&	0.2 &	 \multicolumn{4}{c|}{$V_{base}$ = 12.35 KV}		 	\\

7	&	8	&	0.015	&	0.015	&	17	&	18	&	0.122	 &	0.092	&	43	&	45	&	0.061	&	 0.015	&	 26	&	0.07 &	 \multicolumn{4}{c|}{$S_{base}$ = 1000 KVA} 	 	\\

8	&	35	&	0.046	&	0.015	&	17	&	20	&	0.214	 &	0.046	&	 	&	 	&	 	&	 	&	27	&	0.13		 &	 \multicolumn{4}{c|}{$Z_{base}$ = 152.52 $\Omega$}   	 	 \\	
	
\hline
\end{tabular}
\label{data}
\end{table*}

We now provide some numerical examples to illustrate the difference between the convergence conditions and rates of the three algorithms based on piecewise linear droop control functions (\ref{eq:plf}). The network topology (Fig. \ref{42bus}) and parameters (TABLE \ref{data}) are based on a distribution feeder of South California Edison.  
As shown in Fig. \ref{42bus}, Bus 1 is the actual ``0" bus, and five PVs are installed on Bus 2, 12, 26, 29, and 31 respectively.\footnote{Unlike what is implied in the system model and its analysis, in practice we may not have control at all buses. As a result, the convergence conditions $C1$, $C2$ and $C3$ need to be modified accordingly, based on an ``effective'' reactance matrix that takes into consideration non-control buses.} AC power flow model is applied in our simulation, calculated with MatLab package MatPower\cite{zimmerman1997matpower}, instead of the linear model we use in the analytical characterization. 

The deadband for control function is chosen to be $[0.98^{p.u.},1.02^{p.u.}]$ for all buses, and the hard voltage threshold $\overline{v}_i$ and $\underline{v}_i$ in control function is designed to be a variable, adjusted for the purpose of comparison of convergence conditions by $\alpha_i=q_i^{max}/(\overline{v}_i-\delta/2)$. 

\subsubsection{Convergence condition} We start with observing the difference among the convergence conditions of the three algorithms.
\begin{itemize}
	\item We first present in Fig.\ref{allconverge} that, once we design control functions and stepsize such that, convergence conditions for $C2$ and $C3$ are met, the dynamical systems $D2$ and $D3$ converge monotonically to the same equilibrium. However, the dynamical system $D1$ converges but may not monotonically.
	\item We then change the slope of the control function such that we have a larger $A^{-1}$ (i.e., smaller $\alpha_i$). This will give $D2$ a more strict condition, and $D3$ a less strict one. Resultantly, as we see in Fig.\ref{convcon}(a), $D2$ no longer converge. However, by simply decreasing stepsize $\gamma_g$, $D2$ can be brought back to convergence, as shown in Fig. \ref{convcon}(b).
	\item Lastly, we change the slope of control function to get a smaller $A^{-1}$ (i.e., larger $\alpha_i$). This affects the convergence conditions for $D1$ and $D3$, while leaving that for $D2$ inviolated. Similarly, $D3$ can be back to convergence by having a smaller stepsize. This is shown in Fig. \ref{convcon}(c-d).
\end{itemize}

\begin{figure}[t]
	\begin{center}		
		\includegraphics[trim = 30mm 80mm 0mm 85mm, clip, scale=0.6]{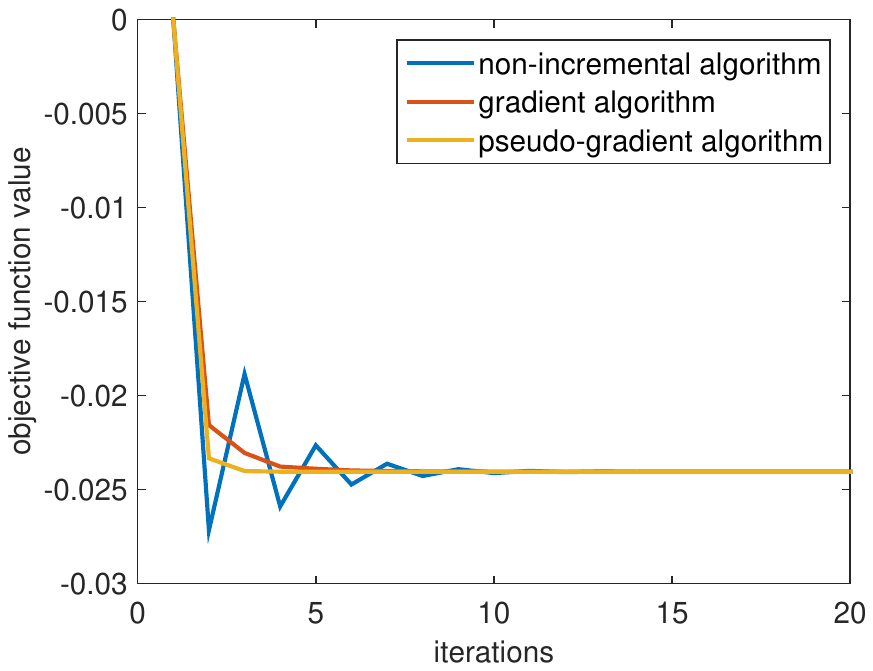}
		\caption{$D2$ and $D3$ both converge} \label{allconverge}
		\includegraphics[trim = 30mm 80mm 0mm 80mm, clip, scale=0.6]{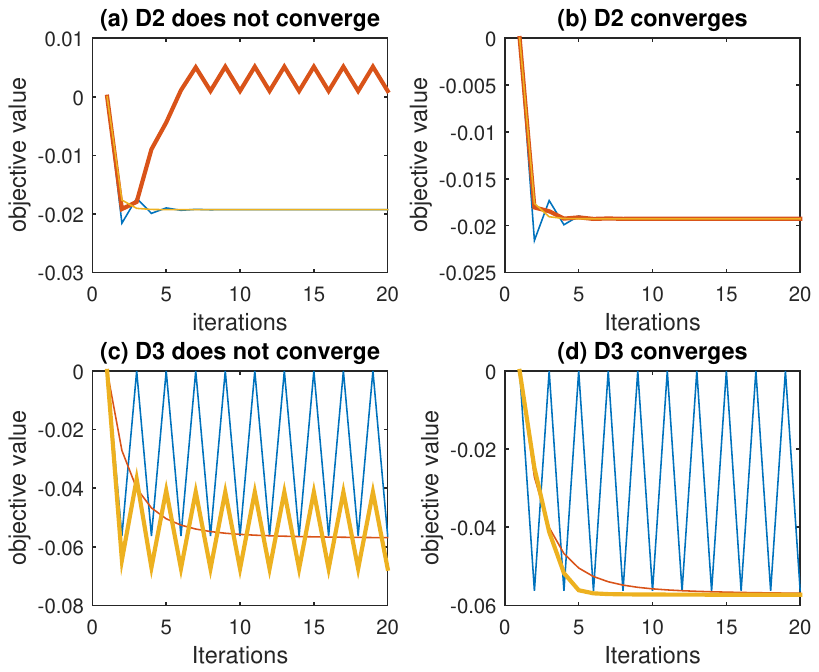}
		\caption{D2 and D3 can be brought back to convergence by changing stepsizes $\gamma_g$ and $\gamma_p$ to small enough values. } \label{convcon}
	\end{center}
\end{figure}

\subsubsection{Range of the stepsize for convergence}
Proposition \ref{prop2} shows that the upper bounds for the stepsizes in $D2$ and $D3$ are related with a factor $\max\{\alpha_i\}$. Since $\max\{\alpha_i\}$ is just a bound, it is interesting to see how tight it is. 
For the linear control function, $\max\{\alpha_i\}=\max(q_i^{max}/(\overline{v}-\delta/2))$, assuming we have universal and symmetric hard voltage threshold $\overline{v}-v^{nom}=v^{nom}-\underline{v}$. We tune $\overline{v}$ such that the value of $\overline{v}-v^{nom}$ ranges from $0.03^{p.u.}$ to $0.18^{p.u.}$ with granularity of $0.01^{p.u.}$, and value of $\max(\alpha_i)$ ranges from 158 to 9.84 accordingly. We examine the largest possible stepsize $\max(\gamma_g)$ and $\max(\gamma_p)$, and compare their ratio with the theoretical convergence boundary factor $\max(\alpha_i)$. The granularity for $\gamma_g$ and $\gamma_p$ is 1 and 0.05 respectively. The results in Fig.\ref{fig_range} illustrates that the simulated relationship of convergence ranges for gradient algorithm and pseudo-gradient algorithm is close to the theoretical one, which serves well as a conservative upper bound. It supports our analysis in Proposition \ref{prop2} that these two algorithms have a one-to-one corresponding convergence ranges for $\gamma_g$ and $\gamma_p$.

\begin{figure}[t]
	\begin{center}		
		\includegraphics[trim = 30mm 80mm 0mm 85mm, clip, scale=0.6]{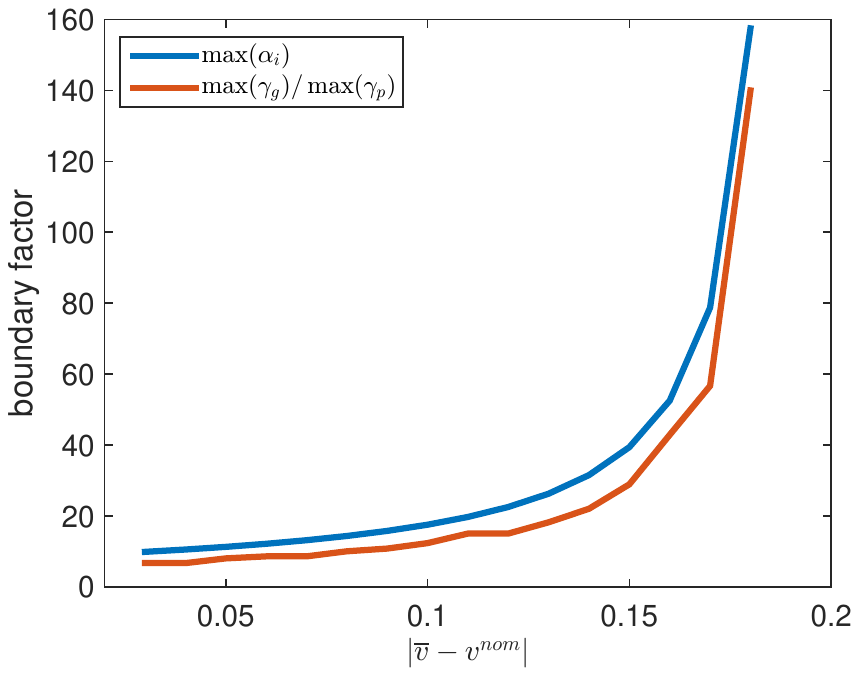}
		\caption{The upper bounds for $\gamma_g$ and $\gamma_p$ is related by a factor close to theoretical value $\max(\alpha_i)$.} \label{fig_range}
	\end{center}
\end{figure}

\subsubsection{Convergence rate}
We observe the convergence rates under certain fixed control functions, with stepsizes $\gamma_g$ and $\gamma_s$ tuned within convergence conditions.\footnote{Since simulations run under non-linear model, the boundary for stepsizes will usually be different from that obtained under linear model. We carefully choose the values of step size so that convergence results are still obtained.} Since $D1$ is a special case for $D3$ with $\gamma_p=1$ and fixed convergence rate, assuming it still fits the convergence condition, we won't specifically involve it. 

We fixed the hard voltage threshold as $\overline{v}_i=0.92^{p.u.}$, and $\underline{v}_i=1.08^{p.u.}$, change the stepsize until it reaches the convergence condition boundary. The results are shown in Fig. \ref{rate_gradient} and \ref{rate_pseugradient}. We can see that, the convergence rates for both gradient algorithm and pseudo-gradient algorithm increase monotonically with the stepsizes before they reach upper bounds, where oscillation of objective function value takes place, i.e., when $\gamma_g=11$, and $\gamma_p=0.6$, and starts bringing down the convergence rates. Also, both algorithms perform similarly in terms of convergence rates, with minimal number of steps less than 10.

\begin{figure}[t]
	\begin{center}		
		\includegraphics[trim = 25mm 80mm 0mm 85mm, clip, scale=0.55]{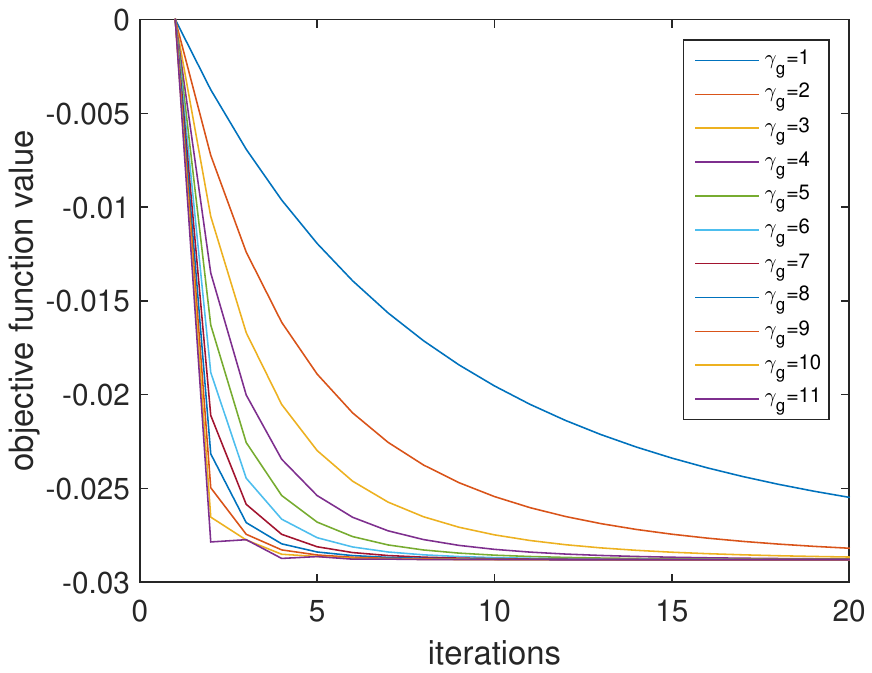}
		\caption{Convergence rates of gradient algorithm with different stepsizes: larger $\gamma_g$ leads to faster convergence rates.} \label{rate_gradient}
		\includegraphics[trim = 25mm 80mm 0mm 85mm, clip, scale=0.55]{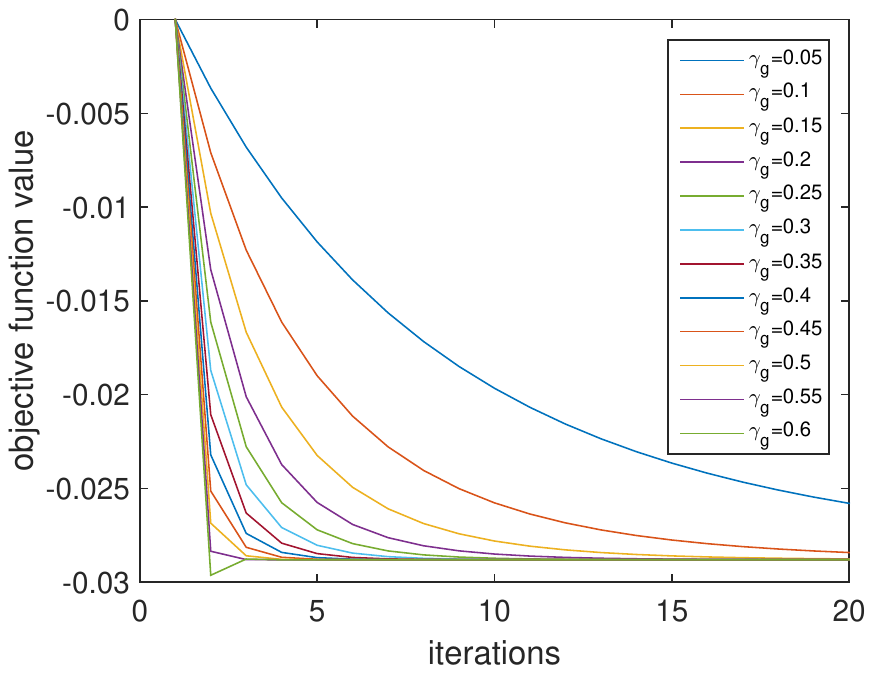}
		\caption{Convergence rates of pseudo-gradient algorithm with different stepsizes: larger $\gamma_p$ leads to faster convergence rates.} \label{rate_pseugradient}
	\end{center}
\end{figure}

\section{conclusion} \label{conclusion}
Motivated by two previously proposed inverter-based local volt/var control algorithms, we have proposed a pseudo-gradient based voltage control algorithm for the distribution network that does not constrain the allowable control functions and has low implementation complexity. We characterize the convergence of the proposed voltage control scheme, and compare it against the two previous algorithms in terms of the convergence condition as well as the convergence rate.


\begin{thebibliography}{9}

\bibitem{Masoud-CDC13}
M. Farivar, L. Chen, and S. Low, ``Equilibrium and dynamics of
local voltage control in distribution systems," { 52nd IEEE Annual Conference on Decision and Control (CDC)}, pp. 4329--4334, 2013.

\bibitem{FZC2015}
M. Farivar, X. Zhou, and L. Chen, ``Local Voltage Control in Distribution System: An Incremental Control Algorithm," {submitted for publication}, 2015.

\bibitem{new1547a}
Standards Coordinating Committee 21 of Institute of Electrical and Electronics Engineers, Inc., IEEE Standard P1547.8\texttrademark /D8, ``Recommended Practice for Establishing Methods and Procedures that Provide Supplemental Support for Implementation Strategies for Expanded Use of IEEE Standard 1547", IEEE ballot document, Aug 2014.

\bibitem{new1547b}
Institute of Electrical and Electronics Engineers, Inc., ``IEEE Standard 1547a\texttrademark (2014)
Standard for Interconnecting Distributed Resources with Electric Power Systems – Amendment 1", May 2014.


\bibitem{Sawin14}
J. L. Sawin, F. Sverrisson, ``Renewables 2014: Global Status Report", Paris: REN21 Secretariat REN21, 2014.


\bibitem{chertkov11} K. Turitsyn, P. Sulc, S. Backhaus, and M. Chertkov, ``Options for control of reactive power by distributed photovoltaic generators,"
{ Proceedings of the IEEE}, 99(6): 1063-"1073, 2011.


\bibitem{smith11} J. Smith, W. Sunderman, R. Dugan, and B. Seal, ``Smart inverter volt/var control functions for high penetration of pv on distribution systems,"  {\em IEEE Power Systems Conference and Exposition (PSCE)}, pp. 1-6, 2011.

\bibitem{Masoud-SGC11}
M.~Farivar, C.~R. Clarke, S.~H. Low, and K.~M. Chandy. ``Inverter {VAR} control for distribution systems with renewables,"
 { IEEE SmartGridComm}, pp. 457--462, 2011.


\bibitem{Garcia13} B. A. Robbins, C. N. Hadjicostis, and A. D. Dominguez-Garcia, ``A two-stage distributed architecture for voltage control in power distribution systems,", { IEEE Trans. on Power Systems}, 28(2): 1470-1482, 2013.


\bibitem{Jahangiri13}
P. Jahangiri, D. C. Aliprantis, ``Distributed Volt/VAr Control by PV Inverters," { Power Systems, IEEE Transactions on }, vol.28, no.3, pp.3429-3439, 2013


\bibitem{Sandia13}
J. Neely, S. Gonzalez, M. Ropp, D. Schutz, ``Accelerating Development of Advanced Inverters: Evaluation of Anti-Islanding Schemes with Grid
Support Functions and Preliminary Laboratory Demonstration," Sandia National Laboratories Technical Report SAND2013-10231, 2013.


%
%
%
%

\bibitem{Lijun-SGC14}
J. Shihadeh, S. You, L. Chen, ``Signal-anticipating in local voltage control in distribution systems," { IEEE SmartGridComm}, pp. 212-217 , 2014.


\bibitem{Kam14}
A. Kam, J. Simonelli, ``Stability of Distributed, Asynchronous VAR-based Closed-loop Voltage Control Systems", { IEEE PES General Meeting Conference \& Exposition}, 2014.

\bibitem{Andren14}
F. Andren, B. Bletterie, S. Kadam, P. Kotsampopoulos, C. Bucher, ``On the Stability of local Voltage Control in Distribution Networks with a High Penetration of Inverter-Based Generation," IEEE Transaction on Industrial Electronics, 62(4): 2519-2529, 2015


\bibitem {Baran89} M.  E.  Baran,  F.  F  Wu,   ``Optimal  Capacitor  Placement  on  radial distribution systems", IEEE Trans. Power Delivery, 4(1):725-734, 1989.

\bibitem{Baran89-2} M. E. Baran, F.F. Wu,  ``Network reconfiguration in distribution systems for loss reduction and load balancing", IEEE Transaction on Power Delivery, 4(2): 401-1407, 1989.

\bibitem{Masoud-TPS13}
M.~Farivar and S.~H. Low.
 ``Branch flow model: relaxations and convexification (parts {I, II})," {\em Power Systems, IEEE Transactions on}, 28(3):2554-2572, 2013.

\bibitem{bertsekas1989paralle}
D. P. Bertsekas, J. N. Tsitsiklis, ``Parallel and distributed computation: numerical methods",  Prentice hall Englewood Cliffs, NJ, 1989.

\bibitem{Khalil2002}
H. K. Khalil and J. W. Grizzle, ``Nonlinear systems'', 3rd Edition, Prentice hall, 200.

\bibitem{zimmerman1997matpower}
Zimmerman, Ray D and Murillo-S{\'a}nchez, Carlos E and Gan, Deqiang, ``MATPOWER: A MATLAB power system simulation package", Manual, Power Systems Engineering Research Center, Ithaca NY, vol. 1, 1997.



\end{thebibliography}
\end{document}